\PassOptionsToPackage{dvipsnames}{xcolor}

\documentclass[sigconf]{aamas}

\usepackage{balance} 
\usepackage{physics}
\usepackage{nicefrac}
\usepackage{subcaption}
\usepackage{tikz, pgfplots}
\usepgfplotslibrary{fillbetween}
\usepackage{wrapfig}

\usepackage{cleveref}
\crefname{algocf}{mech.}{mechs.}
\Crefname{algocf}{Mechanism}{Mechanisms}
\usepackage[linesnumbered,ruled,noend]{algorithm2e}
\SetAlgorithmName{Mechanism}{Mechanism}{Mechanisms}
\allowdisplaybreaks

\setcopyright{ifaamas}
\acmConference[AAMAS '21]{Proc.\@ of the 20th International Conference on Autonomous Agents and Multiagent Systems (AAMAS 2021)}{May 3--7, 2021}{Online}{U.~Endriss, A.~Now\'{e}, F.~Dignum, A.~Lomuscio (eds.)}
\copyrightyear{2021}
\acmYear{2021}
\acmDOI{}
\acmPrice{}
\acmISBN{}

\acmSubmissionID{182}

\title{Timely Information from Prediction Markets}

\author{Grant Schoenebeck}\authornote{Grant Schoenebeck is pleased to acknowledge the support of the National Science Foundation under grants NSF 1618187 and 2007256.}
\affiliation{
  \institution{University of Michigan, Ann Arbor}
  }
\email{schoeneb@umich.edu}

\author{Chenkai Yu}
\affiliation{
  \institution{Tsinghua University, Beijing}
  }
\email{yck17@mails.tsinghua.edu.cn}

\author{Fang-Yi Yu}\authornote{Fang-Yi Yu is pleased to acknowledge the support of the National Science Foundation under grants NSF 1618187, NSF 2007256, CCF-1718549 and IIS-2007887.}
\affiliation{
  \institution{Harvard University}}
\email{fangyiyu@seas.harvard.edu}

\begin{abstract}
Prediction markets are powerful tools to elicit and aggregate beliefs from strategic agents. However, in current prediction markets, agents may exhaust the social welfare by competing to be the first to update the market. We initiate the study of the trade-off between how quickly information is aggregated by the market, and how much this information costs.
We design markets to aggregate timely information from strategic agents to maximize social welfare. To this end, the market must incentivize agents to invest the correct amount of effort to acquire information: quickly enough to be useful, but not faster (and more expensively) than necessary. The market also must ensure that agents report their information truthfully and on time.
We consider two settings: in the first, information is only valuable before a deadline; in the second, the value of information decreases as time passes.
We use both theorems and simulations to demonstrate the mechanisms.
\end{abstract}


\newcommand{\BibTeX}{\rm B\kern-.05em{\sc i\kern-.025em b}\kern-.08em\TeX}

\let\originalleft\left
\let\originalright\right
\renewcommand{\left}{\mathopen{}\mathclose\bgroup\originalleft}
\renewcommand{\right}{\aftergroup\egroup\originalright}

\usepackage{xifthen,bm}

\newcommand{\inner}[2]{\left\langle#1,#2\right\rangle}
\newcommand{\bR}{\mathbb{R}}
\newcommand{\bP}[2][]{\Pr\ifthenelse{\isempty{#1}}{}{_{#1}}\left[#2\right]}
\newcommand{\bE}[2][]{\operatornamewithlimits{\mathbb E}\ifthenelse{\isempty{#1}}{}{_{#1}}\left[#2\right]}
\newcommand{\bI}[2][]{\operatorname{\mathbb I}\ifthenelse{\isempty{#1}}{}{_{#1}}\left[#2\right]}
\newcommand{\Var}[2][]{\mathbf{Var}\ifthenelse{\isempty{#1}}{}{_{#1}}\left[#2\right]}

\newcommand{\given}{\,\middle|\,}
\newcommand{\Exp}[1]{\mathrm{Exp}\mathopen{}\left(#1\right)\mathclose{}}

\newcommand{\cK}{\mathcal{K}}
\newcommand{\cX}{\mathcal{X}}
\newcommand{\cY}{\mathcal{Y}}
\newcommand{\defn}[1]{{\textbf{\textit{#1}}}}

\usepackage{xcolor}

\newboolean{haveappendix}
\setboolean{haveappendix}{true}

\begin{document}


\pagestyle{fancy}
\fancyhead{}


\maketitle 


\section{Introduction}
Eliciting information about uncertain events is crucial for informed decision making.  Information is often acquired by individual agents. To achieve collective intelligence, the key problems are to elicit and aggregate \emph{timely} and truthful information from dispersed agents.

Prediction markets (PMs) allow agents to bet on the occurrence of future events: the outcome of a presidential election, the winner of a football game, etc.  Market prices reflect society's aggregated estimate of the outcome.  However, prediction markets tend to only pay the \emph{first} agent bringing information to the market.  For example, in the market of a tennis match, if one player wins a set, likely the price of the market will shift dramatically.   Because sportscasts are usually delayed by a few seconds, agents with real-time information (say with a confederate attending the match) reap the rewards by trading just seconds before others.  This provides little to no societal value as the price would be updated seconds later anyhow. 
This practice is widespread.  People, called ``courtsiders'', are paid to attend sports events and send back real-time information~\cite{cox2015tennis, dickson2015courtsiding}. The US Open ejected 20 spectators for courtsiding and banned them from future events~\cite{Rothenberg}.  While this information is useful to those profiting from it, it is hardly more than a waste of money and time for our society.  

The importance of speed in trading on information is underscored by the infamous \$300 million, 827 mile fiber optic cable from New York City to Chicago.  The cable reduced the round-trip latency to 13.1 milliseconds~\cite{wired} by offering a more direct route, bypassing Philadelphia, than the previous 1000 mile cable took with round trip latency 14.5 milliseconds.

We want to design systems that work well for society rather than promoting speed that is not needed but is merely a byproduct of the market design.  Increasingly, businesses are seen not merely as profit-maximizing, but as responsible and responsive to their various stakeholders~\cite{roundtable}.



Maximizing welfare rather than profit has various economic motivations as well. To gain market share and maximize long-term revenue, a company may also want to benefit other companies that it makes deals with. For example, many sponsored search auctions maximize welfare instead of revenue~\cite{sponsored_search_auction}. Since profit is a lower bound of social welfare (when agents' utility is non-negative), higher welfare potentially leads to high profit. Welfare-maximizing auctions and profit-maximizing auctions are shown to be very close in terms of both welfare and efficiency~\cite{aggarwal2009efficiency, bulow1994auctions}. In our settings, though our mechanism maximizes social welfare, its profit is still high, illustrated in \Cref{fig:bad_welfare_m}.

Apart from social welfare concerns, another potential challenge with prediction markets is that agents may want to delay reporting their information to increase their rewards~\cite{azar2016should, chen2016informational, kong2018optimizing}.  We show that in the settings we study, this is still a problem for traditional prediction markets, and we resolve this problem in our mechanisms.

\subsection{Our Contribution}
Motivated by the above concerns, we answer the following question: How to aggregate \emph{timely} and truthful information to maximize social welfare? 

We formulate the process as a principal-agent problem.  The principal first suggests a contract which maps agents' reports to rewards.  Then each agent decides his hidden actions (how much effort and how to report) strategically to maximize his utility. The principal knows neither agents' actions nor the relation between agents' actions and quality of information.  This makes our problem different from those in standard contract theory.  To resolve this, we design two new markets to maximize social welfare under the following two settings.

In the single batch setting, a principal needs to decide by a particular deadline.  To maximize social welfare, the principal needs to incentivize agents to invest the correct cost.  Agents' costs are hidden, and they can misreport. For this setting, we propose the Fair Prediction Market (FPM, \Cref{alg:s1}), in which the expected reward for every truthful agent is the same.

In the sequential setting, the value of information decreases as time passes.  Besides dealing with agents' hidden costs and misreporting behaviors, the principal also needs to encourage timely reports. We propose the Marginal Value Prediction (MVP) Market (\Cref{alg:m1}), in which every agent is paid by his contribution to the value of information.

Compared to the traditional prediction market, our mechanisms have more desirable properties, as shown in \Cref{tab:contrib}.

\begin{table}[ht]
  \centering
  \caption{Comparing the traditional prediction market (PM) and our mechanisms. \label{tab:contrib}}
  \begin{tabular}{l c c c}
    \toprule[1pt]
                      & PM              & FPM              & MVP Market \\
    \midrule[0.7pt]
    Timing            & sequential      & single-batch     & sequential \\
    \midrule[0.3pt]
    Truthfulness      & \checkmark      & \checkmark       & \checkmark \\
    \midrule[0.3pt]
    Timeliness        &                 & N.A.             & \checkmark \\
    \midrule[0.3pt]
    Social Optimality &                 & \checkmark       & \checkmark \\
    \bottomrule[1pt]
  \end{tabular}
\end{table}

\subsection{Related Works}
One line of works studies when agents should report their signals in prediction markets~\cite{azar2016should, chen2016informational, kong2018optimizing, gao2013you}.
They find that whether information will be aggregated quickly depends on agents' information structure. Agents will delay reporting if their information is ``complementary'', and rush to report if it is ``substitutional''. Earlier in the finance literature, \cite{kyle1985continuous, holden1992long, foster1996strategic} analyze how private information is disseminated into real financial markets using different models. The market behavior depends on the numbers of insiders, noise traders, and market makers.

\citet{chakraborty2015market} find that when agents are risk-averse, the market scoring rule acts as an opinion pool. Agents' risk aversion avoids the issue that agents will always pull the market price toward their own belief without ever reaching a consensus. We do not have this issue because, in our model, agents believe that the signals of each other are useful and are willing to do Bayesian updates.

For costly information, if the effort level is binary, it's well-known that we can scale the reward to compensate for the cost of effort.  This approach encourages agents to invest effort and increases the liquidity of the prediction~\cite{nisan2007algorithmic}.  In our paper, we consider a more complicated setting: The effort level is continuous, and the relation between agents' actions and the quality of information is unknown.  \citet{azar2018computational} uses contract theory to delegate computation where acquiring data is costly, but assume the principal has some ability to verify the data.   Moreover, our goal is to maximize social welfare, not just collect accurate decisions.

\citet{budish2015high} investigate the continuous limit order book market and the high-frequency trading arms race.  They show examples that such arms races induce rents that harm the liquidity of the market.  They propose a frequent batch market which discretizes the time to mitigate the necessity to be first, and show that the above-mentioned example does not hold in their new markets.

\subsection{Outline.}
In \Cref{section:pre}, we provide some basic notations, assumptions, and definitions; frame the problem we want to solve; and describe how prediction markets work. In \Cref{section:limitation}, we show how prediction markets may fail to collect timely reports. In \Cref{section:s}, we propose \Cref{alg:s1} for the single-batch setting and show that it is truthful and maximizes the social welfare. In \Cref{section:m}, we propose \Cref{alg:m1} for the sequential setting and show that it is truthful, timely and maximizes social welfare.
At last, we present some concrete examples to compare our mechanism to prediction markets in \Cref{section:dyna}.

\section{Preliminaries} \label{section:pre}
There is a \defn{principal} and a set of \defn{agents} $\mathcal{N} = \{1, \dots, n\}$.\footnote{Throughout the paper, we use `she' and `he' for the principal and agents respectively.}  Let $\cY$ be the outcome space, $y \in \cY$ be the true outcome, and $Y$ be the random variable for the outcome.
For each agent $i\in \mathcal{N}$, let $X_i$ be the private information of agent $i$, and $\cX_i$ be the set of possible values of $X_i$ where $\cX_i$ is finite.  The principal wants to collect information from agents to better predict $y$, and her utility depends on the value of information.\footnote{In reality, it could be other people who value the information in the market, use it to do something outside the market, and then get utility based on the quality of the information. Without loss of generality, we simply aggregate all such utility into that of the principal.~\cite{dawid2005geometry,gneiting2007strictly}}

\subsection{Information Structure}
We assume agents' signals $X_1, \dots, X_n$ are i.i.d.\ conditioning on the outcome $Y$. Every agent knows the joint distribution $\bP{X_k, Y}$. 
The principal knows the prior of the outcome $\bP{Y}$,
but she may not know $\bP{X_k \given Y}$.

\begin{example}\label{ex:information}
  Consider a binary outcome space $\cY = \{0, 1\}$, and binary signal spaces $\cX_i = \{0, 1\}$ for all $i \in \mathcal{N}$. The prior of the outcome is given by $\Pr[Y = 1] = \alpha$.  There is a noise level $\beta \in [0,1/2)$ such that each signal is an independent noisy observation of the outcome $\Pr[X_i = Y] = 1 - \beta$ for all $i\in \mathcal{N}$ and $y\in \{0,1\}$.
  Suppose the principal wants to predict $Y$ and her utility is $1$ ($-1$) if correct (incorrect).
  Let $\alpha = 1/2$.  If she only knows the prior of the outcome, her expected utility is $0$. However, if all agents collectively provide their prediction $\bP{Y \given X_1, \dots, X_n}$, the principal's utility is greater than $0$. For example, if $n = 1$, her expected utility is $1-2\beta>0$.
\end{example}

\subsection{The Mechanism Design Problem}\label{section:pre_mechanismdesign}
Acquiring signals is costly for agents.  Hence, agents may not bother to invest effort and may misreport their information.  The principal needs to incentivize agents to invest some effort to acquire their signals, without observing how much effort they actually invest.  Formally, the mechanism has three stages:
\begin{enumerate}
  \item The principal publishes a contract which maps reports and the outcome to payments.  Agents accept (or refuse).
  \item Each agent chooses a hidden effort level and submits a report to the principal.
  \item The true outcome (that agents are guessing at) is revealed.  The principal rewards each agent according to the contract.
\end{enumerate}

Each agent's utility is his reward from the principal minus his effort level.  Agents are rational and maximize their expected utility based on their beliefs over future events.  The principal's utility is the value of information minus the rewards given to the agents.  The \defn{social welfare} is the total utility of all agents and the principal. 

We want to design mechanisms whose resulting social welfare equals that in the \emph{centralized} setting, where every agent's action is controlled by the principal. In particular, we hope that the agents invest the same amount of effort and report just the same as in the centralized setting. 

Throughout the paper, we are only looking for symmetric equilibrium. This is reasonable since agents' signals are i.i.d.\ and they have no a priori means to coordinate. However, it is important to know that a non-symmetric strategy profile can be better than a symmetric one in terms of social welfare. For instance, as shown in \Cref{fig:original}, social welfare may decrease as the number of agents increases, a phenomenon shared across a wide range of economic models: the tragedy of commons, the game of chicken, etc. In such cases, if we break the symmetry and only allow a restricted number of agents to participate, the social welfare will increase.

At each time during the whole process, for each agent, given the information he has and the strategies of other agents, he has a belief over all uncertainty --- including the randomness of the world and other agents' private information. His immediate strategy should maximize his expected utility given his current information. If a strategy profile satisfies the above property, it is called a \defn{perfect Bayesian equilibrium}.\footnote{Here is a more rigorous definition. In an extensive-form game, for each game state $h$, let $\pi^\sigma(h)$ denote its reach probability according to strategy profile $\sigma$. We simply use $\pi(h)$ if $\sigma$ is clear from the context. For each information set $I$, let $\pi(I) = \sum_{h \in I} \pi(h)$. Let $u^\sigma(h)$ denote $h$'s expected utility according to $\sigma$. For each information set $I$ with $\pi(I) > 0$, let $u(I) = \bE[h \in I]{u(h)} = \sum_{h \in I} u(h) \pi(h) / \pi(I)$. For a strategy profile $\sigma$, let $\sigma_{I \to s}$ denote the same strategy profile except the strategy at information set $I$ is changed to $s$. A strategy profile $\sigma = (\sigma_1, \dots, \sigma_n)$ is a \defn{perfect Bayesian equilibrium} if for every information set $I$ with $\pi(I) > 0$, for every strategy $s$ at $I$, $u^\sigma(I) \ge u^{\sigma_{I \to s}}(I)$.}

We say a mechanism is \defn{individually rational} if every agent has non-negative expected utility in every Nash equilibrium.

\subsection{Prediction Markets with Scoring Rules}
Scoring rules have a very long history~\cite{de1937prevision, brier1951verification, good1992rational, gneiting2007strictly}.  Market scoring rules were introduced by \citet{hanson2003combinatorial} to study prediction markets.
A \emph{scoring rule} for an outcome $y$ is a function $S:\Delta_\cY\times\cY\to \mathbb{R}$,\footnote{$\Delta_\cY$ is the set of all probability distributions over $\cY$.} so that $S(p, y)$ is the score assigned to a prediction $p$ when the true outcome realized is $Y = y$.  Formally,
\begin{definition}[Proper Scoring Rule]
  $S: \Delta_\cY \times \cY \to \bR$ is called a \defn{proper scoring rule} if for any $b, p \in \Delta_\cY$, \[\bE[Y \sim b]{S(b, Y)} \ge \bE[Y \sim b]{S(p, Y)}.\]
  A proper scoring rule is \emph{strict} if the equality holds only if $b = p$.
  In other words, reporting one's belief results in a higher score than reporting other distributions. Such a report is said to be \emph{truthful}.
\end{definition}

Prediction markets with a scoring rule proceed as follows: A public belief $p \in \Delta_\cY$ is maintained in the market. Initially, $p = p_0$. Agents can change the market belief arbitrarily, resulting in a belief sequence $p_0, p_1, p_2, \dots$. After the outcome is revealed, for every $k \ge 1$, the agent who changes $p_{k-1}$ to $p_k$ is rewarded (or charged, if negative) by \[S(p_k, y) - S(p_{k-1}, y),\] where $S$ is a proper scoring rule.

In prediction markets, if each agent is only allowed to report once (and they believe other agents are rational), then the reports will be truthful. In particular, each agent will report his posterior distribution using the market information as the prior and his private information as the evidence.

\begin{proposition} \label{prop:update}
  If signals are independent conditioned on the outcome, then agents can perform a Bayesian update knowing only the current market belief (and the likelihood of their own signals).  In particular, they need not know the history of updates.
\end{proposition}

\section{Limitations of Prediction Markets} \label{section:limitation}
In this section, we show by examples how the original form of prediction markets may perform undesirably: agents may 1) invest too much effort, which decreases the social welfare, and 2) intentionally delay their reports.

\subsection{Inflated Effort and Poor Welfare}\label{section:rush}

Current theories of prediction markets do not consider agents' effort to discover signals: How much effort should an agent invest? For some easy, accurate information, it would be a waste of resources if everyone invests a lot. 

We model the cost of information by an \defn{access function} $F: \bR_{\ge 0} \to [0, 1]$ that maps from an effort level (the cost that an agent spends) to the probability of getting signals. Each agent $i$ decides his effort level $c_i$ and then obtains his signal $X_i$ with probability $F(c_i)$. We assume the access function is the same for every agent as common knowledge but unknown to the principal.

Suppose agents' signals are structured as in \Cref{ex:information} with $\alpha = 1/2$ and $\beta = 0$, i.e., all $X_i$'s are identical to the outcome and are exact substitutes to each other. Also, only the first agent who changes the belief will get one unit of reward, and others will get zero.
In this case, every agent wants to be the first. If multiple agents get signals, we assume each of them receives the reward with equal probability since they are symmetric.

Suppose the principal's value of information is described by a proper scoring rule $S$. She needs to choose a proper scoring rule $\tilde{S}$ for the prediction market. Unfortunately, she is unable to maximize the social welfare because she doesn't know the access function $F$ and the information structure $\bP{X \given Y}$ (characterized by $\beta$ in this case). We will see in this section that, if she chooses $\tilde{S} = S$ for prediction markets, the social welfare could be very poor. This is in stark contrast to the mechanisms we propose later in this paper, where the social welfare is maximized when $\tilde{S} = S$, without the need to know anything about $F$ or $\bP{X \given Y}$. Our results are summarized in \Cref{fig:original}.

\begin{figure}[ht]
  \centering
  \begin{subfigure}{0.49\hsize}
    \includegraphics[width=\textwidth]{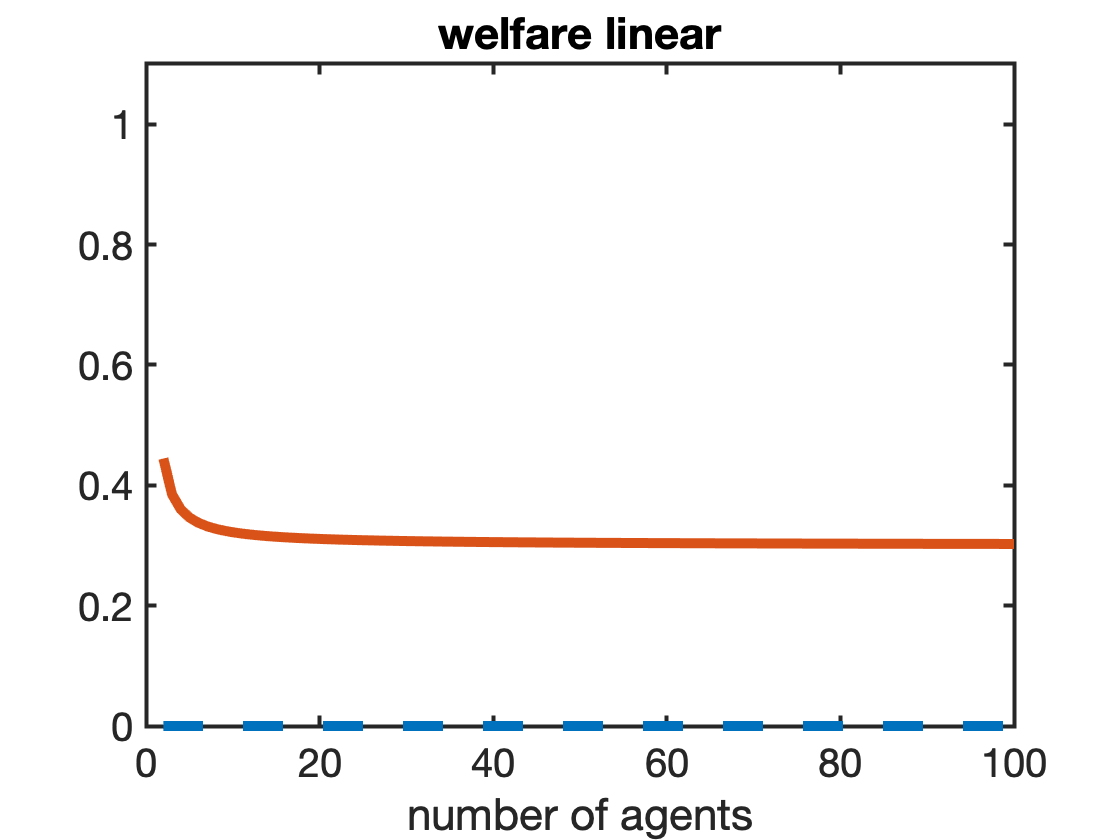}
  \end{subfigure}
  \begin{subfigure}{0.49\hsize}
    \includegraphics[width=\textwidth]{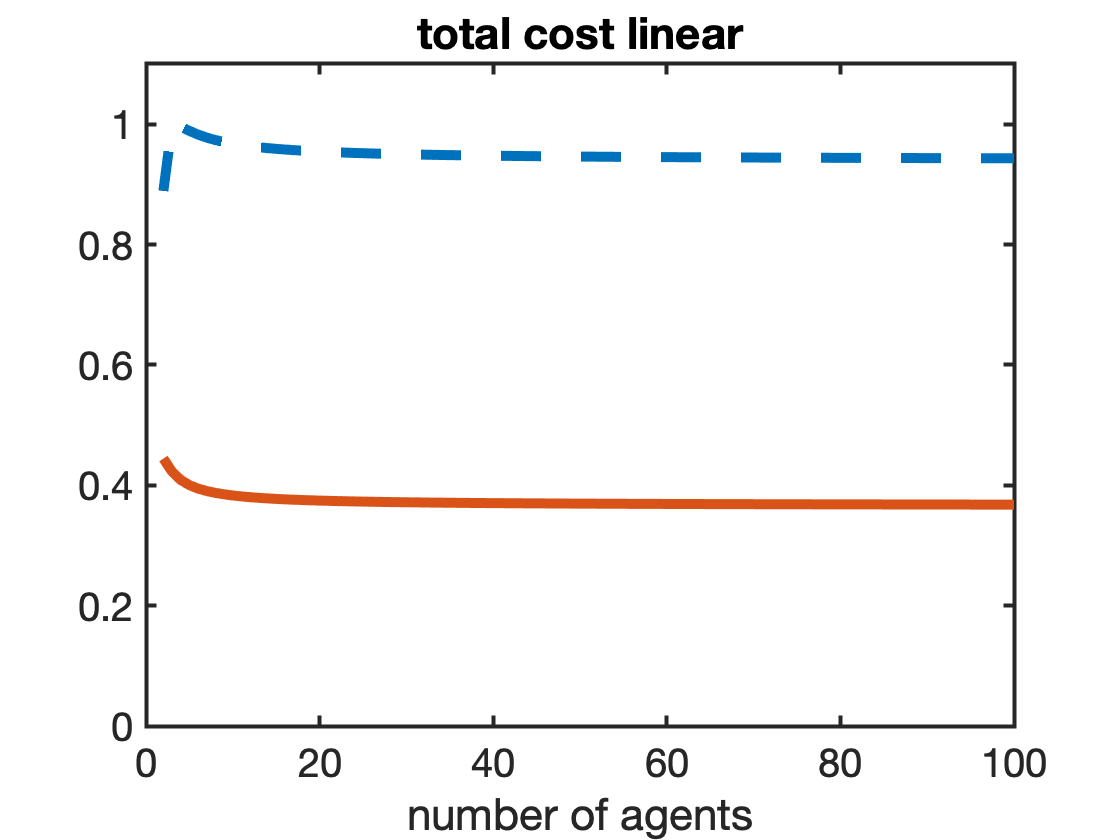}
  \end{subfigure}
  \vskip\baselineskip
  \begin{subfigure}{0.49\hsize}
    \includegraphics[width=\textwidth]{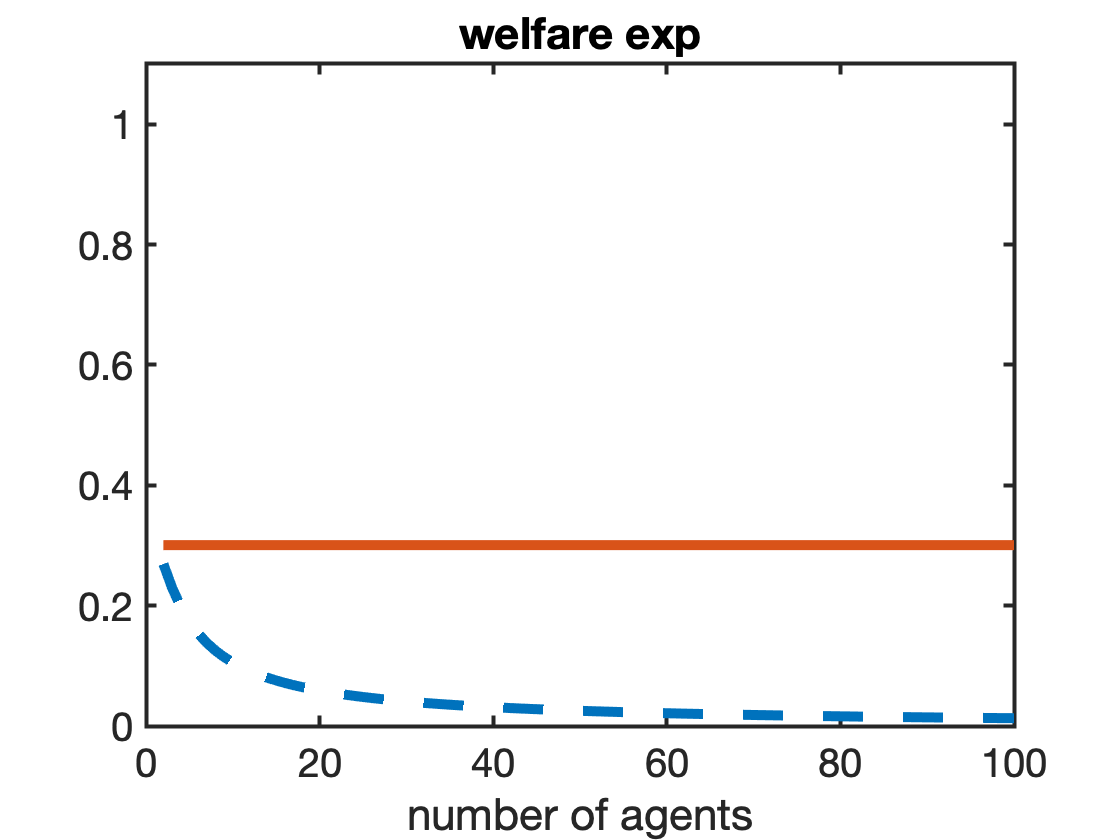}
  \end{subfigure}
  \begin{subfigure}{0.49\hsize}
    \includegraphics[width=\textwidth]{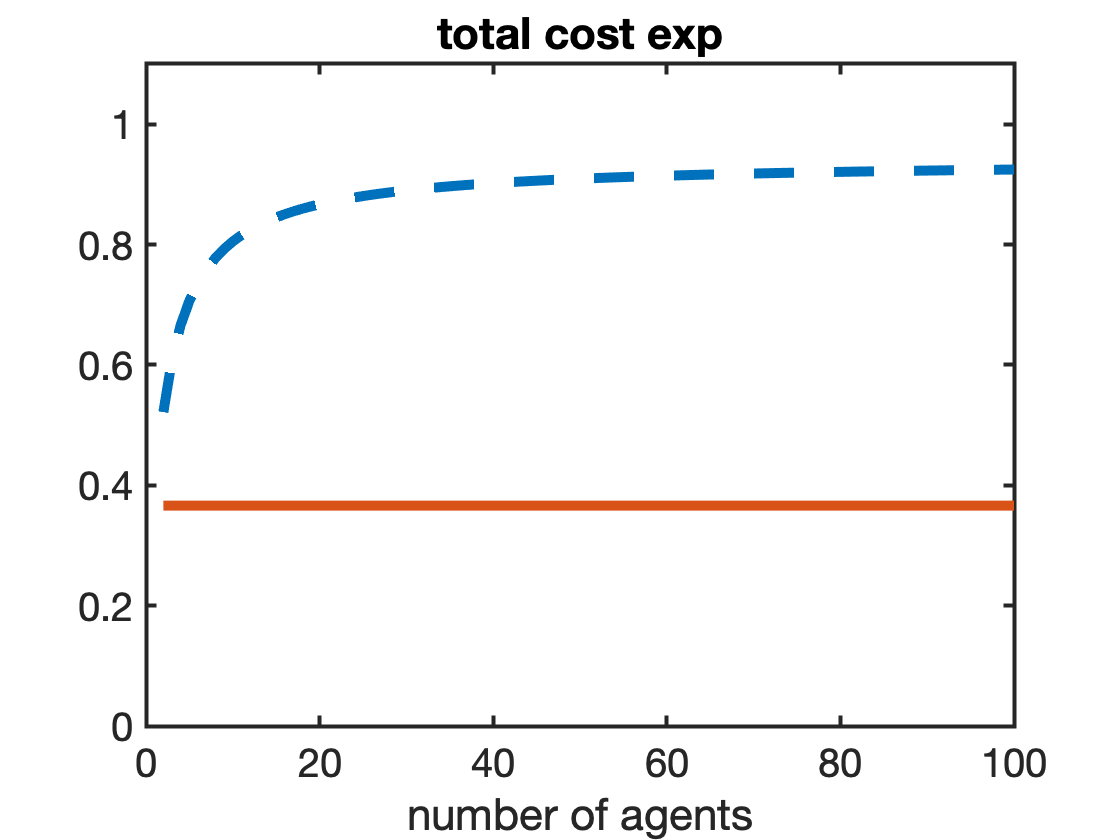}
  \end{subfigure}
  \caption{Red solid: Optimal (Centralized). Blue dashed: Prediction markets. Top: Linear access function $F(c) = 3c$ with $c \in [0,1/3]$. Bottom: Exponential access function $F(c) = 1 - e^{-3c}$ with $c \ge 0$. In both (linear and exponential) cases, the total costs (right column) are high in the strategic setting. In the linear case, agents spend all potential social value and have zero social welfare.  In the exponential case, the total cost increases as the number of agents increases.}
  \label{fig:original}
\end{figure}

We consider two examples of access functions and compute the social welfare of the market as the number of agents increases.  
When every agent's effort is $c$, social welfare is the total value minus cost. The value is $1$ when at least one agent derives his signal, and the total cost is $c n$, so in expectation,
\begin{equation} \label{eq:original_walfare}
  W^{(n)}(c) = \left(1-(1-F(c))^{n}\right) - c n.
\end{equation}

\begin{proposition}[Linear Access] \label{prop:linear}
  Given $n\ge 2$, $\lambda > 1$ and a linear access function is $F(c) = \lambda c$ for $c \in [0, 1 / \lambda]$, the optimal social welfare with $n$ agents is $W^{(n)}(c_{\rm opt}) = 1 - \lambda^{-n/(n-1)} - n \big(\lambda^{-1}-\lambda^{-n/(n-1)}\big)>0$ in the centralized setting, but the social welfare $W^{(n)}(c_{\rm self})$ is zero in the strategic setting.
\end{proposition}
\begin{proposition}[Exponential Access] \label{prop:exp}
  Given $n\ge 2$, $\lambda>1$ and an exponential access function $F(c) = 1 - e^{-\lambda c}$ for $c\ge 0$.  As $n\to \infty$, the optimal social welfare with $n$ agents is $W^{(n)}(c_{\rm opt}) = 1-(1+\ln \lambda )/\lambda>0$, but the social welfare in the strategic setting is $W^{(n)}(c_{\rm self}) = O(1/n)$.
\end{proposition}


\subsection{Delayed Report} \label{section:delay}
Unfortunately, prediction markets do not guarantee timely reports, i.e., one may wait for others to report. This is undesirable when the value of information decays quickly.
Suppose agents' signals are structured as in \Cref{ex:information} with $\alpha = 0.02,\ \beta = 0.2$.  Consider the quadratic scoring rule, where we let $G(p) = \norm{p}_2^2$ and $S(p, y) = G(p) + \inner{\nabla G(p)}{\delta_y - p} = 2 p(y) - \norm{p}_2^2$. Contrary to the intuition, the marginal value of a report does not monotonically decrease. As shown in \Cref{fig:late}, the largest increase in the scoring rule is due to the third report, not the first or the second. Thus, an agent who believes he is likely to be the first will wait before reporting.
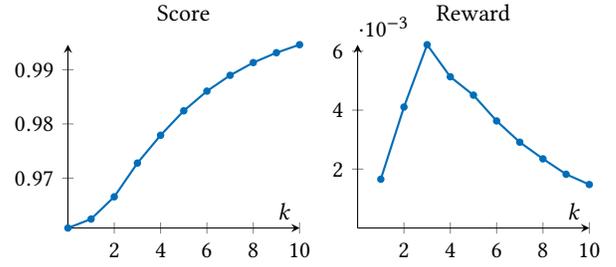
\begin{figure}[ht]
\centering
\begin{tikzpicture}
  \begin{axis}[
      width=0.55\hsize,
      title=Score,
      xlabel=$k$,
      axis lines=middle,
    ]
    \addplot[NavyBlue, thick, mark=*, mark size=1] coordinates {(0, 0.9608) (1, 0.962456) (2, 0.96656) (3, 0.972778) (4, 0.977909) (5, 0.982417) (6, 0.98605) (7, 0.988961) (8, 0.991309) (9, 0.993133) (10, 0.994609)};
  \end{axis}
\end{tikzpicture}
\begin{tikzpicture}
  \begin{axis}[
      width=0.55\hsize,
      title=Reward,
      xlabel=$k$,
      axis lines=middle,
      before end axis/.code={\addplot [draw=none, forget plot] coordinates {(0,0)};}
    ]
    \addplot[NavyBlue, thick, mark=*, mark size=1] coordinates {(1, 0.0016557) (2, 0.00410387) (3, 0.00621856) (4, 0.00513086) (5, 0.00450783) (6, 0.00363284) (7, 0.0029115) (8, 0.00234772) (9, 0.00182409) (10, 0.00147597)};
  \end{axis}
\end{tikzpicture}
\caption{Left: The expected score after $k$ reports. Right: The expected reward for the $k$-th report, which is the difference between two consecutive values in the left plot. Note that it's not decreasing!}
\label{fig:late}
\end{figure}

\section{Single Batch Model} \label{section:s}

We want to design new mechanisms to deal with the above issues. In particular, we assume the market belief is used by someone whose utility is the ``quality'' of the market belief minus the rewards she gives to the agents. Then, we maximize social welfare and encourage truthful and timely reports.
We first ignore the time factor and consider a simple case, where each agent $i$:
\begin{enumerate}
  \item chooses to invest $c_i \in \bR$ much effort, \label{s_step1}
  \item gets a signal $X_i$ with probability $F(c_i)$, and reports $b_i$ to the mechanism, \label{s_step3}
  \item receives some reward $r_i$ from the mechanism. \label{s_step4}
\end{enumerate}
We want to design a reward function (contract) from agents' reports to rewards in order to maximize the social welfare.  If $p$ is the aggregated belief from agents' reports, the value of information is represented as a (strictly) proper scoring rule $S(p, y)$, where $y$ is the outcome.  The principal's utility $U$, each agent $i$'s utility $u_i$, and the social welfare $W$, are given by:
$U = \bE{S(p, Y)} - \sum_i \bE{r_i},\ u_i = \bE{r_i} - c_i,\ W = \bE{S(p, Y)} - \sum_i c_i$.
Here the expectation is taken over all randomness (ex-ante), i.e., agents compute it based on the information in Stage~\ref{s_step1}. Note that given the information valuation $S$ and the information structure, assuming agents are truthful, the ex-ante social welfare only depends on the agents' effort $\bm c = (c_1, \ldots, c_n)$.  We call $\bm c^*$ an \emph{optimal effort profile} if
\begin{equation}\label{eq:righteffort1}
  \bm c^* \in \arg\max_{\bm c} \Big(\bE[p, Y]{S(p, Y)} - \sum_i c_i\Big),
\end{equation}
where $p$ is the Bayesian posterior of the outcome given all signals.

\subsection{Mechanism and Theorem}
To incentivize agents to invest the optimal effort, we want to design a mechanism that, given agents' reports, outputs an aggregated belief $p \in \Delta_\cY$ and a reward for each agent. This task is challenging for two reasons:
\begin{itemize}
  \item The agents' efforts $c_1, \dots, c_n$ and reports $b_1, \dots, b_n$ are decided by each agent individually.
  \item The joint distribution $\bP{X_k, Y}$ and the effort function $F(c)$ are not known to the principal, and thus na\"ively eliciting agents' signals $X_k$ does not work.
\end{itemize}


Our mechanism is shown in \Cref{alg:s1}.
Each agent $k$ is asked to report $b_{k, y} = \bP{X_k \given Y = y}$ for each $y \in \mathcal{Y}$ (one of them can be omitted).
If the mechanism knows $\frac{\bP{Y = y \given X_1, \dots, X_{k-1}}}{\bP{Y \neq y \given X_1, \dots, X_{k-1}}}$, then it is easy to updated it to $\frac{\bP{Y = y \given X_1, \dots, X_k}}{\bP{Y \neq y \given X_1, \dots, X_k}}$ given $b_{k, y}$ because
$\frac{\bP{Y = y \given X_1, \dots, X_k}}{\bP{Y \neq y \given X_1, \dots, X_k}} = \frac{\bP{Y = y \given X_1, \dots, X_{k-1}}}{\bP{Y \neq y \given X_1, \dots, X_{k-1}}} \frac{b_{k, y}}{1 - b_{k, y}}$,
which follows by applying Bayes to both the numerator and denominator of the first two fractions, and then using the fact that $X_k$ is conditionally independent of $X_1, \dots, X_{k-1}$.
As a result, we can update $\bP{Y = y \given X_1, \dots, X_{k-1}}$ to $\bP{Y = y \given X_1, \dots, X_k}$ because we can first compute $\frac{\bP{Y = y \given X_1, \dots, X_{k-1}}}{\bP{Y \neq y \given X_1, \dots, X_{k-1}}}$, use this and $b_{k, y}$ to compute  $\frac{\bP{Y = y \given X_1, \dots, X_k}}{\bP{Y \neq y \given X_1, \dots, X_k}}$, and then transform this back to  $\bP{Y = y \given X_1, \dots, X_k}$.
We succinctly denote the above process as $p_{k, y} = \mathit{Update}(p_{k-1, y}, b_{k, y})$ where
\begin{align} \label{eq:update}
  \mathit{Update}(p_{k, y}, b_{k, y}) = \dfrac{p_{k,y} b_{k,y}}{(1 - p_{k,y}) (1 - b_{k, y}) + p_{k,y} b_{k,y}}
\end{align}
and $p_{k, y} = \bP{Y = y \given X_1, \dots, X_k}$.

\begin{algorithm}[ht]
  \caption{Fair Prediction Market}
  \label{alg:s1}
  \DontPrintSemicolon
  \KwIn{a report profile $(b_1, \dots, b_n)$ where $b_k = (b_{k,1}, \dots, b_{k,d-1})$ is the information provided by agent $k$, describing what the agent $k$'s (claimed) values of $\bP{X_k \given Y = y}$ for each $y$. For those who do not obtain a signal, we assume their $b_{k,y} = \nicefrac{1}{2}$ for $y = 1, \dots, d - 1$, where $d = \abs{\cY}$.}
  \KwOut{the reward $r_k$ for each agent $k$, and the aggregated belief $p_n$}
  \For{$k = 1$ \KwTo $n$}{
    Let $(\pi_1, \dots, \pi_n)$ be a random permutation with $\pi_n = k$\;
    $p_0 \gets \bP{Y}$\;
    \For{$j = 1$ \KwTo $n$}{
      \For(\tcp*[f]{This loop goes through every $y \in \cY$ and update the corresponding entry of $p$ according to the information provided by agent $\pi_j$.}){$y = 1$ \KwTo $d - 1$}{
        $p_{j,y} \gets \mathit{Update}(p_{j-1,y}, b_{\pi_j,y})$ as defined in \eqref{eq:update}
      }
      $p_{j,d} \gets 1 - \sum_{y=1}^{d - 1} p_{j,y}$
    }
    $r_k \gets S(p_n, y^*) - S(p_{n-1}, y^*)$ \tcp*{$y^*$ is the true outcome}
  }
\end{algorithm}

A report is said to be \emph{truthful} if it results in a Bayesian update on the market belief, where the prior is the previous market belief, the posterior is the new market belief, and the evidence is the reporter's signal.
We pay each agent by his improvement on the market belief as if he were the last one to update.
This mechanism is ``fair'' in the sense that if everyone is truthful then everyone receives the same expected reward.
\begin{proposition} \label{prop:z}
  Agent $k$ makes a correct Bayesian update iff $b_{k,y} = \frac{\bP{X_k \given Y = y}}{\bP{X_k \given Y \ne y}}$ for all $y \in \cY$.
\end{proposition}
Note that if $p_{k,y}$ is known for $\abs{\cY} - 1$ different $y$'s, then the last one follows directly since their sum is 1, so each agent needs to report only $\abs{\cY} - 1$ values, the same number as in the original prediction market, where each agent reports a probability distribution.

\begin{theorem} \label{thm:s}
  Assume $\frac{\dd^2 F(c)}{\dd c^2} < 0$.\footnote{decreasing marginal benefit, a very common assumption in economics.}
  \Cref{alg:s1} is individually rational, and there exists a strict perfect Bayesian equilibrium $\bm{\sigma}$ in which the expected social welfare is maximized (over all symmetric strategy profiles), and $\bm \sigma$ satisfies the following properties:
  \begin{description}
    \item[Effort Optimality] The effort profile $\bm c$ is optimal (as in \eqref{eq:righteffort1}).
    \item[Truthfulness] Each agent makes a Bayesian update on the market belief.
  \end{description}
\end{theorem}


\subsection{Proof Sketch}

In order to proof \Cref{thm:s}, we first show some lemmas.

\begin{lemma}[Truthfulness] \label{lemma:s_truthful}
  Every report $b_k$ will be truthful, assuming other reports are truthful. Any deviation will result in a strictly worse expected reward.
\end{lemma}

\begin{lemma}[Effort Optimality] \label{lemma:s_effort}
  Assume $\frac{\dd^2 F(c)}{\dd c^2} < 0$. Agents are incentivized to invest the ``right'' amount of effort \eqref{eq:righteffort1} that maximizes the expected social welfare, assuming all reports are truthful.
\end{lemma}

\begin{lemma} \label{lemma:v_num}
  The expected score $\bE{S(p, Y)}$ of distribution $p$ only depends on the number of previous updates but not who have made updates, assuming all reports are truthful.
\end{lemma}

\begin{proof}[Proof of \Cref{thm:s}]
  Each agent $i$ makes two decisions, and his strategy can be written as $\sigma_i = (c_i, b_i)$, where $c_i$ is the effort he invests at the beginning and $b_i$ is his report. By \Cref{lemma:s_effort}, he will not deviate from $c_i$. By \Cref{lemma:s_truthful}, he will not deviate from $b_i$.
  
  The proof for individual rationality is simple: If in a Nash equilibrium, agent $i$ gets negative expected utility, then he can deviate to $c_i = 0$ and get zero utility. This means he is not in a Nash equilibrium. Thus in every Nash equilibrium, every agent has non-negative expected utility.
\end{proof}

Below is the proofs of the lemmas. \Cref{lemma:s_effort} is perhaps the most interesting among the three.

\begin{proof}[Proof of \Cref{lemma:s_truthful}]
  Assuming other reports are truthful, for avery agent $k$, we have $p_{n-1} = \bP{Y \given \text{all signals except agent $k$'s}}$. By the property of strictly proper scoring rule, his best strategy is to make $p_n = \bP{Y \given \text{all signals}}$ for every possible $p_{n-1}$. Any deviation will lower his reward. This is achievable due to \Cref{prop:update}.
\end{proof}

\begin{proof}[Proof of \Cref{lemma:s_effort}]
  Since agents are symmetric to each other, we look for a symmetric equilibrium, where every agent $i$ invests the same amount of effort $c_i = c$. Let $v_k = \bE{S(p_k, Y) - S(p_0, Y)}$, the expected increase of the score after $k$ updates.

  The expected social welfare is given by:
  \[W = \bE{S(p_k, Y)} - n c = \bE{v_k} + \bE{S(p_0, Y)} - n c.\]
  Setting the derivative to be zero, we have
  \begin{align}
    0 & = \frac{\dd W}{\dd c} = \frac{\dd}{\dd c} \bE{v_k} - n \notag \\
    & = \frac{\dd F(c)}{\dd c} \frac{\dd}{\dd F(c)} \sum_{k=0}^n \binom{n}{k} F(c)^k (1 - F(c))^{n-k} v_k - n \notag \\
    & = \frac{\dd F(c)}{\dd c} \sum_{k=0}^n \binom{n}{k} \Big(k F(c)^{k-1} (1 - F(c))^{n-k} \notag \\
    & \hspace{2cm} - (n - k) F(c)^k (1 - F(c))^{n-k-1}\Big) v_k - n \notag \\
    & = \frac{\dd F(c)}{\dd c} n \Bigg(\underbrace{\sum_{k=1}^n \binom{n - 1}{k - 1} F(c)^{k-1} (1 - F(c))^{n-k} v_k}_{\text{substitute $k$ by $k + 1$}} \notag \\
    & \hspace{2cm} - \sum_{k=0}^{n-1} \binom{n - 1}{k} F(c)^k (1 - F(c))^{n-k-1} v_k\Bigg) - n \notag \\
    & = n \frac{\dd F(c)}{\dd c} \sum_{k=0}^{n-1} \binom{n - 1}{k} F(c)^k (1 - F(c))^{n-k-1} (v_{k+1} - v_k) - n. \label{eq:sw1}
  \end{align}

  On the other hand, the expected utility of agent $i$ is
  \begin{align*}
    u_i
     & = \bP{\text{get signal}} \bE{\text{reward} \given \text{agent $i$ reports}} - \text{cost}            \\
     & = F(c_i) \bE[k, X_1, \dots, X_k, Y]{S(p_k, Y) - S(p_{k-1}, Y) \given \text{agent $i$ reports}} - c_i \\
     & = F(c_i) \bE[k]{v_k - v_{k - 1} \given \text{agent $i$ reports}} - c_i                               \\
     & = F(c_i) \sum_{k=0}^{n-1} \binom{n - 1}{k} F(c)^k (1 - F(c))^{n-k-1} (v_{k+1} - v_k) - c_i,
  \end{align*}
  where $c$ is the amount of effort invested by other agents. Because $F(c_i)$ is concave in $c_i$, $u_i$ is also concave in $c_i$. We also know that $F(c_i)$ is upper bounded (by 1), so $\lim_{c_i \to \infty} \dd F(c_i) / \dd c_i = 0$ and thus $\lim_{c_i \to \infty} \dd u_i / \dd c_i = -1$. Therefore, there is a unique $c_i$ that maximizes the utility, which is either (a) $c_i = 0$ or (b) the point with zero derivative. In case (b), we have
  \begin{equation}
    0 = \frac{\dd u_i}{\dd c_i} = \frac{\dd F(c_i)}{\dd c_i} \sum_{k=0}^{n-1} \binom{n - 1}{k} F(c)^k (1 - F(c))^{n-k-1} (v_{k+1} - v_k) - 1. \label{eq:br1}
  \end{equation}
  \Cref{eq:br1} describes how an agent's decision $c_i$ should best response to those of others $c$. In a symmetric equilibrium, $c_i = c$. Then, surprisingly, \Cref{eq:br1} becomes equivalent to \Cref{eq:sw1}. In case (a), both individually optimal and socially optimal solutions are $c_i = 0$.
  In other words, the distributed maximization of each agent's utility can result in the maximization of social welfare.
\end{proof}

\begin{proof}[Proof of \Cref{lemma:v_num}]
  Since signals are identically distributed and the effort function is the same for every agent, an agent is indistinguishable from another. The lemma simply follows.
\end{proof}

\section{Sequential Model} \label{section:m}
In this section, we consider a setting that involves time.
Its difference from the setting of the previous section is that signals are not received by agents immediately but will be eventually. Formally, before receiving signal $X_i$, each agent $i$ suffers from a latency $T_i$, which is a random variable with c.d.f.\ $F_{c}(t)$. For instance, $F_{c}(t) = 1 - e^{-\lambda c t}$ means $T_i \sim \Exp{\lambda c_i}$. Here, $F_c(t)$ --- a generalization of the access function $F(c)$ used in the previous section --- depends on time.  In summary, each agent $i\in \mathcal{N}$:
\begin{enumerate}
  \item chooses to invest $c_i \in \bR$ effort,
  \item obtains a signal $X_i$ at time $T_i\ge 0$ generated from  c.d.f $F_{c_i}(\cdot)$, decides a time $s_i\ge 0$ and a report $b_i$ to send to the mechanism at time $s_i$,
  \item receives some reward $r_i$ from the mechanism.
\end{enumerate}
We assume that as long as the agents invest non-zero effort, they always obtain their signals before the true outcome being revealed because this is far in the future.
The value of the information (market belief) evolves over time. Let $p(t) \in \Delta_\cY$ denote the market belief at time $t$. The value of a belief history $\{p(t)\}_{t>0}$ is defined as:
$V = \int_{t > 0} S(p(t), y) h(t) \dd t$,
where $y$ is the outcome, $S$ is a strictly proper scoring rule characterizing the quality of the market belief, and \defn{time value function} $h$ is a function characterizing how the value of information diminishes through time. For instance, $h(t) = \eta e^{-\eta t}$ means that the value of information decays exponentially. This would be appropriate if the principal needs to make decision at a random time $\tau \ge 0$ generated from an exponential distribution with parameter $\eta$.

We want to design a mechanism that takes agents' online reports as inputs, maintains a real-time market belief, and finally outputs the reward given to each agent. We also want this mechanism to be truthful, \emph{timely}, and social-welfare-maximizing. We say a mechanism is timely if every agent reports immediately after he gets a signal.
The principal's utility $U$, each agent $i$'s utility $u_i$, and the social welfare $W$, are given by:
$U = \bE{V} - \sum_i \bE{r_i},\ u_i = \bE{r_i} - c_i,\ W = \bE{V} - \sum_i c_i$.
Given proper scoring rule $S$ and the information structure, assuming the agents aggregate their information in a \emph{truthful} and \emph{timely} manner, the expected social welfare only depends on the agents' effort $\bm{c} = (c_1, \ldots, c_n)$.  In this section, we call $\bm{c}^*$ an \emph{optimal effort profile} if
\begin{equation}\label{eq:righteffort2}
  \bm{c}^* \in \arg\max_{\bm c} \Big(\bE[p, Y]{V} - \sum_i c_i\Big).
\end{equation}

\subsection{Mechanism and Theorem}

Besides the challenges involved in \Cref{section:s} --- including the hidden effort, the unknown information structure and effort function, and the potential manipulation of agents' reports --- we also need to deal with another complexity: An agent can choose any time to report, not necessarily just at the time he receives his signal, and even before it (i.e., $s_i < t_i$).
We restrict our focus to mechanisms where each agent can report only once. Note that this assumption exists in the previous literature, e.g., in the traditional prediction market, there would be no truthfulness guarantee without this assumption.
For this sequential setting, we propose \Cref{alg:m1}, which updates a market belief using agents' reports one by one. Agents report their information in the same structure as in \Cref{section:s}.
The mechanism also computes counterfactual market beliefs with one of the reports skipped. In particular, the counterfactual belief for agent $i$'s absence is what the market belief would be if agent $i$ does not report.
The reward for each agent depends on both the actual and counterfactual market beliefs.

\begin{algorithm}[htb]
  \caption{Marginal Value Prediction (MVP) Market}
  \label{alg:m1}
  \DontPrintSemicolon
  \KwIn{online~report~$b_j = (b_{j,1}, \dots, b_{j,d-1})$~from~each~agent~$j$}
  \KwOut{real-time market belief $p_k$, and the reward $r_i$ for each agent $i$}
  $p_0 \gets \bP{Y}$\;
  \For{$i = 1$ \KwTo $n$}{
    $\tilde p^i_0 \gets \bP{Y}$ \tcp*{$\tilde p^i$ is the counterfactual market belief for agent $i$'s absence}
  }
  \For{$k = 1$ \KwTo $n$}{
    wait until receive a report from an agent $j$ and denote it by $b_j$\;
    $t_j \gets$ current time\;
    \For(\tcp*[f]{Update the market belief with agent $j$'s report}){$y = 1$ \KwTo $d - 1$}{
      $p_{k,y} \gets \mathit{Update}(p_{k-1,y}, b_{j,y})$ as defined in \eqref{eq:update}
    }
    $p_{k,d} \gets 1 - \sum_{y=1}^{d-1} p_{k,y}$\;
    \For{$i = 1$ \KwTo $n$}{
      \eIf{$i = j$}{
        $\tilde{p}^i_k \gets \tilde{p}^i_{k-1}$
      }(\tcp*[f]{Update the counterfactual belief for $i$'s absence with $j$'s report}){
        \For{$y = 1$ \KwTo $d - 1$}{
          $\tilde p^i_{k,y} \gets \mathit{Update}(\tilde p^i_{k-1,y}, b_{j,y})$ as defined in \eqref{eq:update}
        }
        $\tilde p^i_{k,d} \gets 1 - \sum_{y=1}^{d-1} \tilde p^i_{k,y}$
      }
    }
  }
  \tcp{After the true outcome $y^*$ reveals}
  Let $k(t) = \sum_{j = 1}^n \bI{t_j < t}$, $p(t) = p_{k(t)}$, and $\tilde{p}^i(t) = \tilde{p}^i_{k(t)}$.\;
  \For{$i = 1$ \KwTo $n$}{
  $r_i \gets \int_{t>0} (S(p(t), y^*) - S(\tilde{p}^i(t), y^*)) h(t) \dd t$ \label{line:m1_r}
  }
\end{algorithm}

\begin{theorem} \label{thm:m1}
  Assume $h(t) > 0$ for all $t > 0$, $\int_{t>0} h(t) \dd t < \infty$, $\frac{\dd^2 F_c(t)}{\dd c^2} < 0$,\footnote{decreasing marginal benefit} and $F_c(t)$ is the c.d.f. of a non-negative random variable for $c \ge 0$.
  \Cref{alg:m1} is individually rational, and there exists a strict perfect Bayesian equilibrium that is socially optimal (over all symmetric strategy profiles) and satisfies the following properties:
  \begin{description}
    \item[Effort Optimality] Every agent invests the ``right'' amount of effort as \eqref{eq:righteffort2}.
    \item[Truthfulness] Each agent makes a Bayesian update on the market belief.
    \item[Timeliness] For all $i\in \mathcal{N}$, $s_i = t_i$.
  \end{description}
\end{theorem}

\subsection{Intuition and Proof Sketch}

Our core idea is to pay each agent by the actual value of information minus the counterfactual value of information as if he had not \emph{updated the market belief}. Let $\tilde{V}^i$ be the counterfactual value w.r.t.\ agent $i$. The reward (in Line~\ref{line:m1_r} of \Cref{alg:m1}) is given by
\[V - \tilde{V}^i = \int_{t>0} (S(p(t), y) - S(\tilde{p}^i(t), y)) h(t) \dd t,\]
where $\tilde{p}^i(t)$ is what the market belief would be at time $t$ if agent $i$ had not changed anything in the market. \Cref{fig:ex-ante-reward} gives an intuition for $\bE{V - \tilde{V}^i}$.
Note that we are not talking about a counterfactual value for the case as if an agent had not \emph{participated the game}. The number of agents is still $n$, and other agents do the same.

\begin{figure}[t]
  \centering
  \begin{tikzpicture}
    \newcommand{\enormous}{\fontsize{17}{17}\selectfont}
    \node[violet!50] at (1.7 * 0.7, 1.1 * 0.7) {\textsf{\textbf{\enormous r}}};
    \node[violet!50] at (2.7 * 0.7, 1.1 * 0.7) {\textsf{\textbf{\enormous e}}};
    \node[violet!50] at (3.7 * 0.7, 2.3 * 0.7) {\textsf{\textbf{\enormous w}}};
    \node[violet!50] at (4.4 * 0.7, 3.4 * 0.7) {\textsf{\textbf{\enormous a}}};
    \node[violet!50] at (5.4 * 0.7, 4.2 * 0.7) {\textsf{\textbf{\enormous r}}};
    \node[violet!50] at (6.5 * 0.7, 5.05 * 0.7) {\textsf{\textbf{\enormous d}}};
    \begin{axis}[
        scale=0.7,
        transform shape,
        xlabel=$t$,
        ylabel=$v_{k(t)}$,
        xtick=\empty,
        ytick=\empty,
        extra x ticks = {1.16225},
        extra x tick labels = {$t^*$},
        ymax=0.8,
        no markers,
        axis lines=middle,
        legend style={at={(0.02, 0.85)},anchor=north west}
      ]
      \addplot+[RedOrange, thick, name path = cf] coordinates {(0, 0) (0.313178, 0) (0.313178, 0.0816327) (3.86526,
          0.0816327) (3.86526, 0.2399) (4.78407, 0.2399) (4.78407,
          0.403947) (5.58132, 0.403947) (5.58132, 0.547144) (7.05955,
          0.547144) (7.05955, 0.662703) (8.05955, 0.662703)};
      \addplot+[NavyBlue, thick, name path = a] coordinates {(0, 0) (0.313178, 0) (0.313178, 0.0816327) (1.16225,
          0.0816327) (1.16225, 0.2399) (3.86526, 0.2399) (3.86526,
          0.403947) (4.78407, 0.403947) (4.78407, 0.547144) (5.58132,
          0.547144) (5.58132, 0.662703) (7.05955, 0.662703) (7.05955,
          0.752018) (8.05955, 0.752018)};
      \addplot[gray, dashed] coordinates {(1.16225, 0) (1.16225, 0.0816327)};
      \addplot fill between[of = cf and a, split, every even segment/.style = {violet!10}, every odd segment/.style = {violet!10}];
      \legend{c.f.\ value, actual value}
    \end{axis}
  \end{tikzpicture}
  \caption{\emph{Ex-ante} total reward received by an agent (taken expectation over \emph{all} agents' signals) for a fixed time sequence of signal discovery. $v_k = \bE{S(p_k, y) - S(p_0, y)}$ is the expected increase of score due to the first $k$ reports, $k(t)$ is the number of report up to time $t$, and $t^*$ is the time of report. A late report (larger $t^*$) reduces the reward. A non-truthful report shifts the \emph{actual curve} downward, also reducing reward.
  }
  \label{fig:ex-ante-reward}
\end{figure}
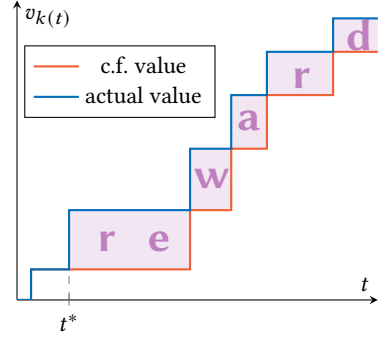



\begin{lemma}[Effort Optimality] \label{lemma:right_effort}
  Agents are incentivized to invest the ``right'' amount of effort that maximizes the expected social welfare, assuming all updates are timely and truthful.
\end{lemma}

\begin{lemma}[Truthfulness] \label{lemma:truthful}
  No matter what time an agent makes his update, a truthful update is better than a non-truthful one, assuming all other updates are truthful.
\end{lemma}

\begin{lemma}[Timeliness] \label{lemma:timely}
  Every agent is incentivized to update the market belief as soon as he gets his signal, assuming all updates are truthful.
\end{lemma}

\begin{proof}[Proof of \Cref{thm:m1}]
  Each agent $i$ decides how much effort to make at the beginning. By \Cref{lemma:right_effort}, he will not deviate at this decision. Then, at each time, he decides whether to report and if so, what to report. By \Cref{lemma:timely}, he will report if he has a signal. By \Cref{lemma:truthful}, he will report truthfully.
  %
\end{proof}

The proofs of the lemmas are similar to those in \Cref{section:s} and are postponed to the supplementary materials.

\subsection{Connection to VCG}
There are both similarities and differences between our mechanisms and the VCG mechanism. They are similar because both of them have a payoff function that can be interpreted as an actual term minus a counterfactual term. Also, in our mechanisms, the utility function of an agent is --- to some extent --- aligned with the social welfare as a function of his action, as in VCG.

However, a straightforward application of VCG fails. In VCG, we need to compute the utility of each agent, which is impossible here, because the amount of effort each agent invests is never revealed. In our mechanisms, the alignment of the agent's utility and social welfare is achieved implicitly without the principal computing them. In addition, VCG deals with a single-stage game, while our mechanisms deal with multi-stage games.  This is to say that the signals must be discovered before they can be (truthfully) revealed. Finally, VCG guarantees the DSIC (dominant-strategy incentive-compatible) property, which is not the case in our setting, where agents respond to others when choosing the effort level.

\begin{figure*}[t]
\centering
\begin{minipage}[t]{.32\textwidth}
  \centering
  \begin{tikzpicture}
    \begin{axis}[
        width=\textwidth,
        xlabel={$\lambda$, ease},
        ylabel={$c$, effort},
        ymax=0.41,
        xmin=0,
        no markers,
        axis lines=middle,
        legend style={nodes={scale=0.7, transform shape}},
      ]
      \addplot+[thick, NavyBlue, dashed, domain=0:15] {0.25};
      \addplot+[thick, RedOrange] coordinates {(0.5, 0.) (0.55, 0.0596281) (0.6, 0.107634) (0.65,
          0.146831) (0.7, 0.179214) (0.75, 0.206231) (0.8,
          0.228959) (0.85, 0.248214) (0.9, 0.264624) (0.95, 0.278681) (1.,
          0.290773) (1.05, 0.301213) (1.1, 0.310252) (1.15,
          0.318096) (1.2, 0.324917) (1.25, 0.330854) (1.3,
          0.336026) (1.35, 0.340531) (1.4, 0.344454) (1.45,
          0.347867) (1.5, 0.35083) (1.55, 0.353397) (1.6, 0.355613) (1.65,
          0.357516) (1.7, 0.359143) (1.75, 0.360521) (1.8,
          0.361679) (1.85, 0.362638) (1.9, 0.363419) (1.95, 0.364041) (2.,
          0.364519) (2., 0.364519) (2.1, 0.365099) (2.2, 0.365256) (2.3,
          0.365066) (2.4, 0.364591) (2.5, 0.363883) (2.6, 0.362982) (2.7,
          0.361922) (2.8, 0.360733) (2.9, 0.359436) (3., 0.358051) (3.,
          0.358051) (3.5, 0.350304) (4., 0.341978) (4.5, 0.333658) (5.,
          0.325601) (5.5, 0.317917) (6., 0.310643) (6.5, 0.303779) (7.,
          0.297311) (7.5, 0.291216) (8., 0.285468) (8.5, 0.280042) (9.,
          0.274912) (9.5, 0.270055) (10., 0.265449) (10.5, 0.261076) (11.,
          0.256917) (11.5, 0.252956) (12., 0.249178) (12.5,
          0.245569) (13., 0.242118) (13.5, 0.238814) (14.,
          0.235647) (14.5, 0.232607) (15., 0.229687)};
      \legend{pred.\ market, MVP market};
    \end{axis}
  \end{tikzpicture}
  \caption{How ease affects agent's effort.}
  \label{fig:eas}
\end{minipage}
\hfill
\begin{minipage}[t]{.32\textwidth}
  \centering
  \begin{tikzpicture}
    \begin{axis}[
        width=\textwidth,
        xlabel={$\beta$, noise},
        ylabel={$c$, effort},
        ymin=0,
        xmax=0.39,
        xtick={0.1, 0.2, 0.3},
        no markers,
        axis lines=middle,
        legend style={nodes={scale=0.7, transform shape}},
      ]
      \addplot+[thick, NavyBlue, dashed] coordinates {(0., 0.9) (0.005, 0.813105) (0.01, 0.733842) (0.015,
          0.661436) (0.02, 0.595219) (0.025, 0.534609) (0.03,
          0.479097) (0.035, 0.428236) (0.04, 0.38163) (0.045,
          0.338929) (0.05, 0.299819) (0.055, 0.264018) (0.06,
          0.231271) (0.065, 0.20135) (0.07, 0.174045) (0.075,
          0.149163) (0.08, 0.126531) (0.085, 0.105985) (0.09,
          0.0873764) (0.095, 0.0705662) (0.1, 0.0554253) (0.105,
          0.0418332) (0.11, 0.0296771) (0.115, 0.0188513) (0.12,
          0.00925666) (0.125, 0.000799779) (0.13, -0.00660716)};
      \addplot+[thick, RedOrange] coordinates {(0., 0.341641) (0.005, 0.313929) (0.01, 0.286163) (0.015,
          0.258298) (0.02, 0.230291) (0.025, 0.202101) (0.03,
          0.17369) (0.035, 0.145021) (0.04, 0.116058) (0.045,
          0.0867627) (0.05, 0.0570997) (0.055,
          0.0270306) (0.06, -0.00348493)};
      \addplot+[thick, YellowOrange] coordinates {(0., 0.448683) (0.005, 0.437657) (0.01, 0.426356) (0.015,
          0.414765) (0.02, 0.40287) (0.025, 0.390658) (0.03,
          0.378118) (0.035, 0.365237) (0.04, 0.352007) (0.045,
          0.338418) (0.05, 0.324463) (0.055, 0.310133) (0.06,
          0.295423) (0.065, 0.280327) (0.07, 0.264839) (0.075,
          0.248956) (0.08, 0.232674) (0.085, 0.215987) (0.09,
          0.198891) (0.1, 0.163453) (0.12, 0.0873232) (0.14,
          0.00292199) (0.16, -0.0957086)};
      \addplot+[thick, yellow!70!black] coordinates {(0., 0.381056) (0.005, 0.382439) (0.01, 0.383301) (0.015,
          0.383649) (0.02, 0.383491) (0.025, 0.382834) (0.03,
          0.381689) (0.035, 0.380066) (0.04, 0.377975) (0.045,
          0.375426) (0.05, 0.372432) (0.055, 0.369004) (0.06,
          0.365154) (0.065, 0.360893) (0.07, 0.356234) (0.075,
          0.35119) (0.08, 0.345772) (0.085, 0.339994) (0.09,
          0.333868) (0.1, 0.320629) (0.12, 0.290574) (0.14,
          0.256491) (0.16, 0.219273) (0.18, 0.179778) (0.2,
          0.13878) (0.22, 0.0968958) (0.24, 0.0544652) (0.26,
          0.0112018) (0.28, -0.0360081)};
      \addplot+[thick, green] coordinates {(0., 0.232195) (0.005, 0.236822) (0.01, 0.240827) (0.015,
          0.244248) (0.02, 0.247121) (0.025, 0.24948) (0.03,
          0.251358) (0.035, 0.252785) (0.04, 0.253787) (0.045,
          0.254391) (0.05, 0.254618) (0.055, 0.254491) (0.06,
          0.254028) (0.065, 0.253249) (0.07, 0.252169) (0.075,
          0.250805) (0.08, 0.249171) (0.085, 0.247282) (0.09,
          0.245152) (0.1, 0.240215) (0.12, 0.227966) (0.14,
          0.213123) (0.16, 0.196313) (0.18, 0.178104) (0.2,
          0.159007) (0.22, 0.139463) (0.24, 0.119847) (0.26,
          0.100457) (0.28, 0.0815165) (0.3, 0.0631813) (0.32,
          0.0455409) (0.34, 0.0286234) (0.36,
          0.0123897) (0.38, -0.00331423)};
      \legend{pred.\ market,{MVP, $\lambda = 0.5$},{MVP, $\lambda = 1$},{MVP, $\lambda = 3$},{MVP, $\lambda = 12$}};
    \end{axis}
  \end{tikzpicture}
  \caption{How noise affects agent's effort.}
  \label{fig:noise}
\end{minipage}
\hfill
\begin{minipage}[t]{.32\textwidth}
  \centering
  \begin{tikzpicture}
    \begin{axis}[
        width=\textwidth,
        xlabel=$v_1$,
        ylabel={$c$, effort},
        ymax=0.26,
        no markers,
        axis lines=middle,
        legend style={at={(1,0.11)},anchor=south east,{nodes={scale=0.7, transform shape}}},
        ytick={0.1,0.2},
      ]
      \addplot+[thick, NavyBlue, dashed, domain=1:2] {(x - 1) / 2};
      \addplot+[thick, RedOrange] coordinates {(1., 0) (1.04, 0.0184216) (1.08, 0.03435) (1.12,
          0.0484578) (1.16, 0.0611652) (1.2, 0.0727549) (1.24,
          0.0834278) (1.28, 0.0933321) (1.32, 0.102581) (1.36,
          0.111264) (1.4, 0.119451) (1.44, 0.1272) (1.48, 0.134559) (1.52,
          0.141567) (1.56, 0.148259) (1.6, 0.154663) (1.64,
          0.160804) (1.68, 0.166704) (1.72, 0.172383) (1.76,
          0.177855) (1.8, 0.183138) (1.84, 0.188242) (1.88,
          0.193182) (1.92, 0.197966) (1.96, 0.202605) (2., 0.207107)};
      \addplot+[thick, YellowOrange] coordinates {(1., 0.207107) (1.04, 0.20981) (1.08, 0.212372) (1.12,
          0.214808) (1.16, 0.217127) (1.2, 0.219342) (1.24,
          0.221459) (1.28, 0.223488) (1.32, 0.225434) (1.36,
          0.227303) (1.4, 0.229102) (1.44, 0.230835) (1.48,
          0.232506) (1.52, 0.23412) (1.56, 0.235678) (1.6,
          0.237186) (1.64, 0.238646) (1.68, 0.24006) (1.72,
          0.241432) (1.76, 0.242763) (1.8, 0.244055) (1.84,
          0.24531) (1.88, 0.246531) (1.92, 0.247718) (1.96, 0.248874) (2.,
          0.25)};
      \addplot+[thick, yellow!70!black] coordinates {(1., 0.25) (1.04, 0.248901) (1.08, 0.247825) (1.12,
          0.246772) (1.16, 0.245742) (1.2, 0.244734) (1.24,
          0.243747) (1.28, 0.242782) (1.32, 0.241838) (1.36,
          0.240914) (1.4, 0.24001) (1.44, 0.239125) (1.48, 0.23826) (1.52,
          0.237413) (1.56, 0.236584) (1.6, 0.235773) (1.64,
          0.234979) (1.68, 0.234203) (1.72, 0.233443) (1.76,
          0.232699) (1.8, 0.231971) (1.84, 0.231258) (1.88,
          0.23056) (1.92, 0.229877) (1.96, 0.229208) (2., 0.228553)};
      \addplot+[thick, green] coordinates {(1., 0.228553) (1.04, 0.226706) (1.08, 0.224873) (1.12,
          0.223055) (1.16, 0.221252) (1.2, 0.219464) (1.24,
          0.217693) (1.28, 0.215937) (1.32, 0.214198) (1.36,
          0.212477) (1.4, 0.210773) (1.44, 0.209086) (1.48,
          0.207418) (1.52, 0.205768) (1.56, 0.204137) (1.6,
          0.202525) (1.64, 0.200932) (1.68, 0.199359) (1.72,
          0.197806) (1.76, 0.196272) (1.8, 0.194759) (1.84,
          0.193266) (1.88, 0.191793) (1.92, 0.190342) (1.96, 0.18891) (2.,
          0.1875)};
      \legend{pred.\ market,{MVP, $\lambda = 1$},{MVP, $\lambda = 2$},{MVP, $\lambda = 4$},{MVP, $\lambda = 8$}};
    \end{axis}
  \end{tikzpicture}
  \caption{How substitutability $v_1 / v_2$ (proportional to $v_1$ since $v_2$ is fixed) affects agent's effort.}
  \label{fig:subst}
\end{minipage}
\end{figure*}

\section{Simulations} \label{section:dyna}

The equilibrium depends on a variety of parameters. In this section, we analyze how it is affected by ease, noise, and substitutability of the information. We compare our socially optimal mechanism (\Cref{alg:m1}) with the traditional prediction market.\footnote{Following \Cref{section:rush}, we study the case where $\tilde{S} = S$.}
We assume the value of information decays exponentially: $h(t) = \eta e^{-\eta t}$ with parameter $\eta = 1$. The latency of signal discovery is also exponentially distributed: $F_c(t) = 1 - e^{-\lambda c t}$.  Here, $\lambda$ can be viewed as the ease of collecting the information, as the larger $\lambda$ is, the shorter latency the agent suffers.
There are $n = 2$ agents unless otherwise stated.
To simplify the calculation, we let $v_k = \bE{S(p_k, y) - S(p_0, y)}$. It is an intermediate variable that depends on the information structure $\bP{X \given Y}$ (or $\beta$).

\subsection{Ease}
Let $v_1 = 2$ and $v_2 = 3$, defined above. In our mechanism, the amount of effort in equilibrium automatically adapts to the ease of collecting information, even though the mechanism does not know anything about the parameters!
\Cref{fig:eas} is a visualization.
In prediction markets, agents invest too much for very easy ($\lambda \to \infty$) information, thus making the updates unnecessarily quick.
For information that takes a long time to discover, agents still invest equally much effort in prediction markets, while in our mechanism, they do not invest anything because we have the $\eta e^{-\eta t}$ term in the value of information --- the information value decays so quickly that the gain in information value is overwhelmed by the amount of effort invested.

\subsection{Noise}
Consider the scenario described in \Cref{section:delay}, with $\alpha = 0.1$. Recall that $\beta = \bP{X \ne Y}$ is the probability each signal differs from the true outcome. It can also be regarded as the noise of information.
It turns out that in prediction markets, agents invest too much not only for very easy information but also for very accurate ($\beta \to 0$) information, as shown in \Cref{fig:noise}.
When the signals become weak enough, agents no longer invest anything in either prediction markets (because a late reporting is encouraged) or our mechanism (because the gain in value is too little).

\begin{figure}[ht]
  \centering
  \includegraphics[width = 0.98 \hsize]{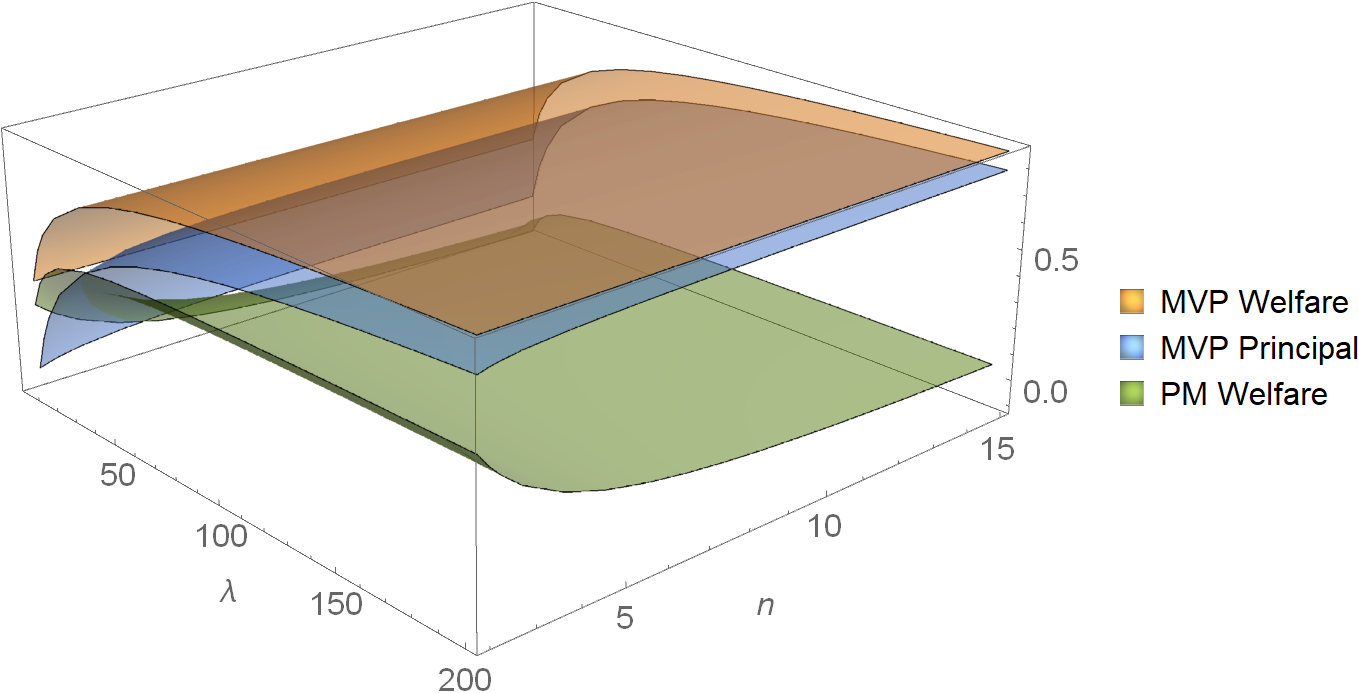}
  \caption{A comparison of the social welfare of the MVP market, the principal's utility of MVP market, and the social welfare of the traditional predictions markets. When there are many agents ($n$ is large) and the signals do not take much time to obtain ($\lambda$ is large), the MVP market has both high social welfare and high utility for the principal.}
  \label{fig:bad_welfare_m}
\end{figure}

\subsection{Substitutability}
Recall that $v_k$ is the expected increase of the score due to the first $k$ updates. We fix $v_2 = 2$, and see what happens when $v_1$ changes. Here, $v_1 / v_2$ can be considered as the substitutability of information.
As shown in \Cref{fig:subst}, in prediction markets, when value is more concentrated in the first report (higher substitutability), agents invest more effort to get a signal quickly.
Surprisingly, this is not always the case in our mechanism. When signals are very easy (quick to obtain), agents tend to invest \emph{less} when value is more concentrated in the first signal. This is because even though the first report brings high value, after the second report, the marginal value of the first report (against other signals) becomes much less. Moreover, the easier to obtain a signal, the stronger this effect.

\subsection{Social Welfare and Principal's Utility}
As we argue that agents invest too much effort for easy and accurate information in the traditional prediction market, one may wonder how bad the effect of such behavior could be on social welfare. We find that the social welfare approaches to 0 as the number of agents grows large, while in our proposed mechanism, both social welfare and the principal's utility are high, as shown in \Cref{fig:bad_welfare_m}. In this example, we assume $v_0 = 0$ and $v_k = 1$ for each $k \ge 1$ (or equivalently, $\beta = 0$). This is similar to what we show in \Cref{section:rush}.





\balance

\bibliographystyle{ACM-Reference-Format} 
\bibliography{ref}


\ifhaveappendix

\clearpage
\appendix

\section{Proof of Proposition~\ref{prop:update}}

\begin{proof}
  We use mathematical induction. Without loss of generality, we relabel the agents by the order they update the market. The agent who makes the $k$-th update has a signal $X_k$ and is given a market belief $p_{k-1}$. Suppose $p_{k-1} = \bP{Y \given X_1, \dots, X_{k-1}}$. Clearly, the base case $p_0 = \bP{Y}$ holds. The Bayesian update for $p_k$ is performed as follows:
  \begin{align*}
    p_k
     & \propto p_{k-1} \bP{X_k \given Y}                                             \\
     & = \bP{Y \given X_1, \dots, X_{k-1}} \bP{X_k \given Y}                         \\
     & \propto \bP{X_1, \dots, X_{k-1} \given Y} \bP{Y} \bP{X_k \given Y}            \\
     & = \bP{X_1, \dots, X_k \given Y} \bP{Y} \text{\quad(conditional independence)} \\
     & \propto \bP{Y \given X_1, \dots, X_k}.
  \end{align*}
  If two probability distributions are proportional, they are the same. Thus, $p_k = \bP{Y \given X_1, \dots, X_k}$.
\end{proof}

\section{Proofs for Propositions in Section~\ref{section:rush}}
Suppose agents are selfish in the above game, and choose their effort level strategically.  If agent $i$ invest effort $x$ and everyone else invests effort $c$, agent $i$'s expected utility is
\begin{equation}\label{eq:original_utility}
  u_i(x;c) = F(x)\sum_{k=0}^{n-1}\frac{1}{k+1}\binom{n-1}{k}F(c)^k(1-F(c))^{n-1-k}-x.
\end{equation}
\begin{proof}[Proof for \Cref{prop:linear}]
  In the centralized setting, the principal can set every agents' to a fixed value.  The maximum of \eqref{eq:original_walfare} happens at $c_{\rm opt} = \lambda^{-1} - \lambda^{-n/(n-1)}$ when $n \ge 2$, and the optimal social welfare is
  $$W^{(n)}(c_{\rm opt}) = 1 - \lambda^{-n/(n-1)} - n \left(\lambda^{-1} - \lambda^{-n/(n-1)}\right)$$
  which converges to $1 - \lambda^{-1} - \lambda^{-1} \ln\lambda>0$ as $n\to \infty$.

  For selfish agents, everyone invests $c_{\rm self}$ is a Bayesian Nash equilibrium when no one wants to deviate.  In our linear effort function case, agent $i$'s expected utility $u(x;c)$ in \eqref{eq:original_utility}  is linear in his effort level $x$.
  Therefore, everyone invests $c_{\rm self}$ is a Bayesian Nash equilibrium when
  $\sum_{k=0}^{n-1}\frac{1}{k+1}\binom{n-1}{k}F(c_{\rm self})^k(1-F(c_{\rm self}))^{n-1-k}=\lambda^{-1}.$
  By some calculation $ 1-(1-F(c_{\rm self}))^{n}=nc_{\rm self}$, so the social welfare~\eqref{eq:original_walfare} under $c_{\rm self}$ is
  $W^{(n)}(c_{\rm self}) = 0$.
  where all agents invest all possible reward to compete to be the first one.  This completes the proof.
\end{proof}

\begin{proof}[Proof of \Cref{prop:exp}]
  Using similar computation, in the centralized setting, the optimal cost is $c_{\rm opt} = \ln \lambda/(n\lambda)$.
  By \Cref{eq:original_walfare} and the definition of exponential effort function, $W^{(n)}(c) = 1-e^{-n\lambda c}-cn$.  Therefore, 
  $$W^{(n)}(c_{\rm opt}) = 1-\frac{1}{\lambda}-\frac{\ln \lambda}{\lambda}.$$
  
  For selfish agents, the symmetric Bayesian Nash equilibrium happens when everyone invests $c_{\rm self}$ and
  \begin{equation}\label{eq:exp1}
      n(1-\exp(-\lambda c_{\rm self})) = \lambda\exp(-\lambda c_{\rm self})(1-\exp(-n\lambda c_{\rm self})).
  \end{equation}
  Therefore, $\exp(-\lambda c_{\rm self}) = 1-\Theta(1/n)$.  Therefore, there exist $0<\alpha_L<\alpha_H$, and a sequence $(\alpha_n)_n$ such that $\exp(-\lambda c_{\rm self}) = 1-\alpha_n/n$, and  $\alpha_n$ are all in an interval $[\alpha_L, \alpha_H]$.  Now we can compute the social welfare in the strategic setting.  Applying \Cref{eq:exp1} to  \Cref{eq:original_walfare}, we have
  \begin{align*}
    W^{(n)}(c_{\rm self})
    & = \frac{n}{\lambda}(\exp(\lambda c_{\rm self})-1)-c_{\rm self}n \\
    & = \frac{n}{\lambda}\frac{\alpha_n}{n}-\frac{\alpha_n}{\lambda}+O\left(\frac{1}{n}\right)\tag{by Taylor expansion} \\
    & = O(1/n).
  \end{align*}
  
  We provide numerical results in \Cref{fig:original}.
\end{proof}

\section{Remaining Proofs of Section~\ref{section:s}}

\begin{proof}[Proof of \Cref{prop:z}]
  Assume w.l.o.g.\ that agents $1, \dots, k - 1$ have made updates and the others have not. A correct Bayesian update should transform $p_y = \bP{Y = y \given X_1, \dots, X_{k-1}}$ into $p'_y = \bP{Y = y \given X_1, \dots, X_k}$. We first calculate the following quantity:
  \begin{align*}
    & \quad \frac{p_y}{1 - p_y} \frac{b_{k,y}}{1 - b_{k, y}} \\
    & = \frac{\bP{Y = y \given X_1, \dots, X_{k-1}}}{\bP{Y \ne y \given X_1, \dots, X_{k-1}}} \frac{\bP{X_k \given Y = y}}{\bP{X_k \given Y \ne y}} \\
    & = \frac{\bP{X_1, \dots, X_{k-1} \given Y = y} \bP{Y = y}}{\bP{X_1, \dots, X_{k-1} \given Y \ne y} \bP{Y \ne y}} \frac{\bP{X_k \given Y = y}}{\bP{X_k \given Y \ne y}} \tag{Bayes rule} \\
    & = \frac{\bP{X_1, \dots, X_k,\ Y = y}}{\bP{X_1, \dots, X_k,\ Y \ne y}} \tag{conditional independence} \\
    & = \frac{\bP{Y = y \given X_1, \dots, X_k}}{\bP{Y \ne y \given X_1, \dots, X_k}}.
  \end{align*}
  Then we have
  \begin{align*}
    p'_y
    = \frac{\frac{p_y}{1 - p_y} \frac{b_{k,y}}{1 - b_{k,y}}}{1 + \frac{p_y}{1 - p_y} \frac{b_{k,y}}{1 - b_{k,y}}}
    & = \frac{\frac{\bP{Y = y \given X_1, \dots, X_k}}{\bP{Y \ne y \given X_1, \dots, X_k}}}{1 + \frac{\bP{Y = y \given X_1, \dots, X_k}}{\bP{Y \ne y \given X_1, \dots, X_k}}} \\
    & = \bP{Y = y \given X_1, \dots, X_k}.
  \end{align*}
\end{proof}

\section{Proofs of Section~\ref{section:m}}

\begin{proof}[Proof of \Cref{lemma:right_effort}]
  Since agents are symmetric to each other, we look for a symmetric equilibrium, where every agent $i$ invests the same amount of effort $c_i = c$. By \Cref{lemma:v_num}, the expected score $\bE{S(p(t), y)}$ at any time $t$ only depends on the number of updates before time $t$, assuming all updates are truthful.
  Let $v_k$ be the expected increase of the score due to the first $k$ updates, and $k(t)$ be the number of updates by time $t$, i.e., the number of values among $t_1, \dots, t_n$ that are less than $t$. Then we have
  \[v_{k(t)} = \bE[X_1, \dots, X_n, Y]{S(p(t), Y) - S(p(0), Y)}.\]
  Let $f_c(t)$ be the p.d.f.\ of the latency distribution, i.e., $f_c(t) = \frac{\dd}{\dd t} F_c(t)$.
  Without loss of generality, we relabel the agents by the order they get their signals.

  The expected social welfare is given by:
  \begin{align*}
    W & = \bE[T_1, \dots, T_n]{V} - n c \\
    & = \bE[T_1, \dots, T_n]{\int_{t > 0} v_{k(t)} h(t) \dd t} \\
    & \hspace{2cm} + \underbrace{\bE[T_1, \dots, T_n]{\int_{t > 0} \bE{S(p(0), y)} h(t) \dd t}}_{\text{constant}} {} - n c.
  \end{align*}
  Setting the derivative to be zero, we have
  \begin{align}
    0 & = \frac{\dd W}{\dd c} = \int_{t>0} \frac{\dd}{\dd c} \bE[T_1, \dots, T_n]{v_{k(t)}} h(t) \dd t - n \notag \\
    & = \int_{t>0} \frac{\dd F_c(t)}{\dd c} \frac{\dd}{\dd F_c(t)} \sum_{k=0}^n \binom{n}{k} F_c(t)^k (1 - F_c(t))^{n-k} v_k h(t) \dd t - n \notag \\
    & = \int_{t>0} \frac{\dd F_c(t)}{\dd c} \sum_{k=0}^n \binom{n}{k} \Big(k F_c(t)^{k-1} (1 - F_c(t))^{n-k} \notag \\
    & \hspace{2cm} - (n - k) F_c(t)^k (1 - F_c(t))^{n-k-1}\Big) v_k h(t) \dd t - n \notag \\
    & = \int_{t>0} \frac{\dd F_c(t)}{\dd c} n \Bigg(\underbrace{\sum_{k=1}^n \binom{n - 1}{k - 1} F_c(t)^{k-1} (1 - F_c(t))^{n-k} v_k}_{\text{substitute $k$ by $k + 1$}} \notag \\
    & \hspace{1.5cm} - \sum_{k=0}^{n-1} \binom{n - 1}{k} F_c(t)^k (1 - F_c(t))^{n-k-1} v_k\Bigg) h(t) \dd t - n \notag \\
    \begin{split}
      & = n \int_{t>0} \frac{\dd F_c(t)}{\dd c} \sum_{k = 0}^{n - 1} \binom{n - 1}{k} F_c(t)^k (1 - F_c(t))^{n - k - 1} \\
      & \hspace{4cm} \cdot (v_{k + 1} - v_k) h(t) \dd t - n.
    \end{split} \label{eq:sw}
  \end{align}

  On the other hand, the expected utility of agent $i$ is given by:
  \[u_i = \int_{s>0} f_{c_i}(s) \bE[T_{-i}, X_1, \dots, X_n, Y]{V - \tilde{V} \given T_i = s} \dd s - c_i,\]
  where $c_i$ is the effort made by agent $i$, and $T_{-i}$ denotes the latency of all agents except $i$ (i.e., $T_1, \dots, T_{i - 1}, T_{i + 1}, \dots, T_n$). Note that the expectation is \emph{ex ante}, i.e., it is computed by the agent at the very beginning before knowing any information. We have
  \begin{align*}
     & \bE[X_1, \dots, X_n, Y]{V - \tilde{V} \given T_i = s} \\
     & = \bE[X_1, \dots, X_n, Y]{\int_{t > 0} (S(p(t), Y) - S(\tilde{p}(t), Y)) h(t) \dd t \given T_i = s} \\
     & = \left.\int_{t > s} (v_{k(t)} - v_{k(t) - 1}) h(t) \dd t \given T_i = s\right..
  \end{align*}
  Then,
  \begin{align*}
    u_i & = \int_{s>0} f_{c_i}(s) \underbrace{\bE[T_{-i}]{\int_{t > s} (v_{k(t)} - v_{k(t) - 1}) h(t) \dd t \given T_i = s}}_{\text{independent to $c_i$}} \dd s - c_i \\
    & = \underbrace{\int_{s>0} f_{c_i}(s) \int_{t > s}}_{\text{switch the order}} \sum_{k = 0}^{n - 1} \binom{n - 1}{k} F_c(t)^k (1 - F_c(t))^{n - k - 1} \\
    & \hspace{4cm} \cdot (v_{k + 1} - v_k) h(t) \dd t \dd s - c_i \\
    & = \int_{t>0} \int_{0<s<t} f_{c_i}(s) \dd s \sum_{k = 0}^{n - 1} \binom{n - 1}{k} F_c(t)^k (1 - F_c(t))^{n - k - 1} \\
    & \hspace{4cm} \cdot (v_{k + 1} - v_k) h(t) \dd t - c_i \\
    & = \int_{t>0} F_{c_i}(s) \sum_{k = 0}^{n - 1} \binom{n - 1}{k} F_c(t)^k (1 - F_c(t))^{n - k - 1} \\
    & \hspace{4cm} \cdot (v_{k + 1} - v_k) h(t) \dd t - c_i.
  \end{align*}
  Since $F_{c_i}(s)$ is concave in $c_i$ for any $s$, we know that $u_i$ is also concave in $c_i$, so there is a single $c_i$ that maximizes $u_i$.
  Setting the derivative to be zero, we have
  \begin{equation}
  \begin{split}
    0 & = \frac{\dd u_i}{\dd c_i} = \int_{t>0} \frac{\dd F_{c_i}(s)}{\dd c_i} \sum_{k = 0}^{n - 1} \binom{n - 1}{k} F_c(t)^k (1 - F_c(t))^{n - k - 1} \\
    & \hspace{4cm} \cdot (v_{k + 1} - v_k) h(t) \dd t - 1. \label{eq:best_response}
  \end{split}
  \end{equation}
  \Cref{eq:best_response} describes how much effort $c_i$ agent $i$ should make in order to best response to others' effort $c$. In a symmetric equilibrium, $c_i = c$, and then \Cref{eq:best_response} become equivalent to \Cref{eq:sw}.
\end{proof}

\begin{proof}[Proof of \Cref{lemma:truthful}]
  A non-truthful report will make all later posterior beliefs wrong. By the property of strictly proper scoring rules, this will lower the expected score at each later time, hence the expected reward.
\end{proof}

\begin{proof}[Proof of \Cref{lemma:timely}]
  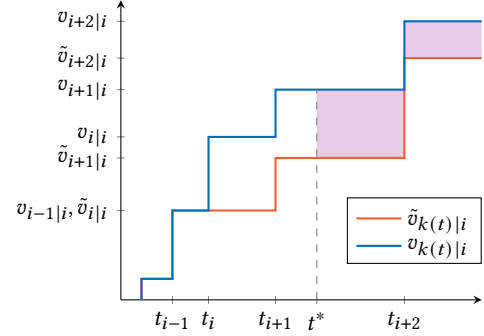
\begin{figure}[ht]
    \centering
    \begin{tikzpicture}[scale=0.7]
      \begin{axis}[
          xtick=\empty,
          ytick=\empty,
          extra x ticks = {1, 1.7, 3, 3.8, 5.5},
          extra x tick labels = {$t_{i - 1}$, $t_i$, $t_{i + 1}$, $t^*$, $t_{i + 2}$},
          extra y ticks = {2, 3, 3.4, 4.3, 4.9, 5.6},
          extra y tick labels = {{$v_{i - 1 | i}, \tilde{v}_{i | i}$}, $\tilde{v}_{i + 1 | i}$, $v_{i | i}$, $v_{i + 1 | i}$, $\tilde{v}_{i + 2 | i}$, $v_{i + 2 | i}$},
          ymin = .3,
          ymax = 6,
          xmin = 0,
          no markers,
          axis lines = middle,
          legend style = {at = {(rel axis cs: 1, 0.12)}, anchor = south east}
        ]
        \addplot+[RedOrange, thick, name path = cf] coordinates {(.4, .3) (.4, .7) (1, .7) (1, 2) (3, 2) (3, 3) (5.5, 3) (5.5, 4.9) (7, 4.9)};
        \addplot+[NavyBlue, thick, name path = a] coordinates {(.4, .3) (.4, .7) (1, .7) (1, 2) (1.7, 2) (1.7, 3.4) (3, 3.4) (3, 4.3) (5.5, 4.3) (5.5, 5.6) (7, 5.6)};
        \addplot[gray, dashed] coordinates {(3.8, 0) (3.8, 4.3)};
        \addplot[violet!20] fill between[of = cf and a, soft clip = {domain=3.8:7}];
        \legend{$\tilde{v}_{k(t) | i}$, $v_{k(t) | i}$}
      \end{axis}
    \end{tikzpicture}
    \caption{\emph{Interim} total reward received by the agents (taken expectation over \emph{future} agents' signals) for a fixed time sequence of signal discovery. The upper curve stands for the actual expected score, the lower curve stands for the counterfactual one, and the area between the two curves and to the right of $t^*$ is the expected reward (\Cref{eq:interim}) if the update is made at time $t^*$, assuming $h(t) = 1$.}
    \label{fig:intuition}
  \end{figure}
  
  Without loss of generality, we relabel the agents by the order they get their signals.
  Let $p_\cK$ with $\cK \subseteq \mathcal{N}$ be the posterior of the outcome $Y$ given the signals of a subset $\cK$ of agents. For example, $p_{\{1, 3\}} = \bP{Y \given X_1, X_3}$. Let \[v_{k | i} = \bE[X_{i + 1}, \dots, X_k, Y]{S(p_{\{1, \dots, k\}}, Y)}\] and \[\tilde{v}_{k | i} = \bE[X_{i + 1}, \dots, X_k, Y]{S(p_{\{1, \dots, i - 1, i + 1, \dots, k\}}, Y)}.\]

  Suppose agent $i$ gets a signal at time $t_i$ and wants to make an update at $t^* > t_i$. Then, his expected reward is
  \begin{align}
    & \quad \bE[X_{i+1}, \dots, X_n, Y]{V - \tilde{V} \given \text{update at } t^*} \notag \\
    & = \bE[X_{i+1}, \dots, X_n, Y]{\int_{t > t^*} (S(p(t), Y) - S(\tilde{p}(t), Y)) h(t) \dd t \given \text{update at } t^*} \notag \\
    & = \int_{t > t^*} (v_{k(t) | i} - \tilde{v}_{k(t) | i}) h(t) \dd t. \label{eq:interim}
  \end{align}

  Note that this expectation is \emph{interim}, i.e., the agent computes it upon getting a signal, knowing the current market belief but not anything in the future.

  \Cref{fig:intuition} shows this with $h(t) = 1$. The later he makes the update, the lower his expected reward is.
  On the other hand, by the property of proper scoring rules, making an update before getting a signal will lead to a non-positive expected reward.
  Thus, his best strategy is to make the update as soon as he gets his signal.
\end{proof}

Note that even if an agent did not report in a timely manner (maybe by mistake), his best strategy is still making the update immediately and truthfully.

\fi

\end{document}